\documentclass[aip,jmp,author-numerical,preprint]{revtex4-1}
\usepackage{amsmath,amsthm,amssymb,amsfonts,mathrsfs,eucal,amsthm}
\usepackage{graphicx,hyperref}
\usepackage[stable]{footmisc}
\providecommand{\U}[1]{\protect\rule{.1in}{.1in}}
\def\sH{\mathcal{H}}
\def\x{\mathsf{x}}

\def\states{\mathfrak{D}}
\def\openone{\mathbb{I}}

\def\reff#1{(\ref{#1})}
\def\N#1{\left|\!\left|{#1}\right|\!\right|}

\newtheorem{theorem}{Theorem}

\newtheorem{lemma}{Lemma}

\def\bi{\begin{itemize}}
\def\ei{\end{itemize}}
\def\be{\begin{equation}}
\def\ee{\end{equation}}
\def\bea{\begin{eqnarray}}
\def\eea{\end{eqnarray}}
\def\ben{\begin{eqnarray*}}
\def\een{\end{eqnarray*}}

\def\eps{\varepsilon}
\def\>{\rangle}
\def\<{\langle}

\def\eps{\varepsilon}

\def\bp{{\bf p}}
\def\bq{{\bf q}}

\newcommand{\bI}{{\mathbf I}}

\def\iden{\mathbb{I}}
\def\bbE{\mathbb{E}}

\def\bbR{\mathbb{R}}

\newcommand{\1} I

\newcommand{\ket}[1]{| #1 \rangle}

\newcommand{\proj}[1]{| #1 \rangle \langle #1 |}

 \DeclareMathOperator{\Span}{Span}
 \DeclareMathOperator{\tr}{Tr}
\DeclareMathOperator{\Tr}{Tr}

\def\*{\star}

\def\tilde{\widetilde}
\def\bar{\overline}

\def\cA{{\cal A}}
\def\cB{{\cal B}}
\def\cC{{\cal C}}
\def\cD{{\cal D}}
\def\cE{{\cal E}}

\def\sH{{\cal H}}

\def\cJ{{\cal J}}
\def\cK{{\cal K}}
\def\cL{{\cal L}}
\def\cM{{\cal M}}
\def\cN{{\cal N}}

\def\cP{{\cal P}}

\def\cS{{\cal S}}
\def\cT{{\cal T}}

\def\cX{{\cal X}}
\def\cY{{\cal Y}}

\def\0{{\mathbf{0}}}
\def\1{{\mathbf{1}}}
\def\2{{\mathbf{2}}}
\def\3{{\mathbf{3}}}
\def\4{{\mathbf{4}}}
\def\5{{\mathbf{5}}}
\def\6{{\mathbf{6}}}
\def\7{{\mathbf{7}}}
\def\8{{\mathbf{8}}}
\def\9{{\mathbf{9}}}
\def\bI{\mathbf{I}}

\def\bp{{\underline{p}}}
\def\bq{{\underline{q}}}

\def\bi{{\bf i}}

\def\bbE{\mathbb{E}}

\def\bbI{{\mathbb{I}}}

\def\bbR{\mathbb{R}}

\begin{document}
\title{Universal coding for transmission of private information}
\author{Nilanjana Datta}   \email{N.Datta@statslab.com.ac.uk}
\affiliation{Statistical Laboratory, University of Cambridge, Wilberforce Road, Cambridge CB3 0WB, UK}

\author{Min-Hsiu Hsieh}\email{minhsiuh@gmail.com}
\affiliation{ERATO-SORST
Quantum Computation and Information Project, Japan Science and Technology
Agency, 5-28-3, Hongo, Bunkyo-ku, Tokyo, Japan 113-0033}
\begin{abstract}
We consider the scenario in which Alice transmits private classical messages to
Bob via a classical-quantum channel, part of whose output is intercepted by an
eavesdropper Eve. We prove the existence of a universal coding scheme under
which Alice's messages can be inferred correctly by Bob, and yet Eve learns
nothing about them. The code is universal in the sense that it does not depend
on specific knowledge of the channel. Prior knowledge of the probability
distribution on the input alphabet of the channel, and bounds on the
corresponding Holevo quantities of the output ensembles at Bob's and Eve's end
suffice.
\end{abstract}

\maketitle

%\begin{IEEEkeywords} Universal coding, Universal covering lemma,
%Universal packing lemma, Universal private coding.
%\end{IEEEkeywords}

%\maketitle

\section{Introduction}

A quantum channel can be used for a variety of different purposes and, unlike
classical channels, it has many different capacities depending on what it is
being used for, on the nature of its inputs and what additional resources are
available to the sender and the receiver. In addition to its use in conveying
classical and quantum information, and generating entanglement, a quantum
channel can also be used to convey {\em{private}} classical information which
is inaccessible to an eavesdropper. This allows unconditionally secure key
distribution, which is impossible in the classical realm.

The different capacities of a quantum channel were first evaluated under the
assumption that the channel was memoryless, that is, correlations in the noise
acting on successive inputs to the channel were assumed to be absent. Holevo
\cite{Hol98} and Schumacher and Westmoreland \cite{SW97} proved that the
classical capacity of a memoryless quantum channel under the restriction of
product-state inputs, is given by the so-called Holevo capacity, the
unrestricted classical capacity then being obtained by a regularisation of this
quantity. An expression for the private classical capacity was independently
obtained by Cai \emph{et al.} \cite{CWY04} and by Devetak \cite{Devetak03}, who
is also credited with the first rigorous proof of the expression for the
quantum capacity (first suggested by Lloyd \cite{Lloyd96} and further justified
by Shor \cite{Shor02}).

An inherent assumption underlying all these results is that the quantum channel
is known {\em{perfectly}} to Alice and Bob. This assumption is, however, not
necessarily valid in real-world communication systems, since it might be
practically impossible to determine all the parameters governing a quantum
channel with infinite accuracy. Thus one often has only limited knowledge of
the quantum channel which is being used. This calls for the design of more
general communication protocols which could be used for transmission of
information through a quantum channel in spite of such a channel uncertainty.
The corresponding coding theorems are then {\em{universal}} in the sense that
they do not rely on exact knowledge of the channel used.

In the quantum setting, progress in this direction was first made by Datta and Dorlas
\cite{Datta-Dorlas07} who obtained an expression for the
classical capacity of a convex combination of memoryless quantum channels. This
corresponds to the case in which Alice and Bob's only prior knowledge is that
the channel in use is one of a given finite set of memoryless channels, with a
given prior probability. It is hence the simplest model with channel
uncertainty. This result was further generalized and extended by Bjelakovic
{\em{et al.}} \cite{BB08ISIT, BBN08PRA} who derived the classical and quantum
capacities of the so-called {\em{compound quantum channels}}, in which the
underlying set of memoryless channels was allowed to be countably infinite or
even uncountable. They also evaluated the optimal rates of entanglement
transmission and entanglement generation through such channels \cite{BBN09CMP}.
In the classical setting, the first study of channel uncertainty dates back to
the work of Wolfowitz \cite{Wolfowitz60,Wolfowitz78}, and of Blackwell {\em{et
al}} \cite{BBT59}, who determined the capacity of compound classical channels.

Note that when Alice is interested in sending only classical messages through a
quantum channel $\Phi$, she first needs to encode her message into a state of a
quantum system which can then be transmitted through the channel. Denoting this
encoding map by ${\cal{E}}$, one effectively obtains a {\em{classical-quantum}}
$(c\rightarrow q)$ channel $W := \Phi \circ {\cal{E}}$ which maps classical
messages into quantum states in the output Hilbert space of the channel $\Phi$.
Hayashi \cite{Hayashi09CMP} proved a universal coding theorem for memoryless
$c\rightarrow q$ channels, as a quantum version of the classical universal
coding by Csisz\'ar and K{\"o}rner \cite{CK81}.

In this paper, we consider transmission of {\em{private}} classical information
through a $c\rightarrow qq$ channel from Alice to Bob and Eve, and prove the
existence of a universal code, for which the private capacity of the channel is
an achievable rate. The channel is defined by the map $W: x \longrightarrow
W^{BE}(x)$ with $x \in {\cal{X}}$ (a finite classical alphabet) and $W^{BE}(x)$
being a state defined on a bipartite quantum system $BE$. Bob has access to the
subsystem $B$, whereas Eve (the eavesdropper) has access to the subsystem $E$.
Such a channel induces two $c\rightarrow q$ channels -- one from Alice to Bob
(which we denote by $W^B$), and one from Alice to Eve (which we denote by
$W^E$). We prove a universal private coding scheme under which Bob can infer
Alice's message with arbitrary precision in the asymptotic limit,
simultaneously ensuring that Eve learns arbitrarily little about the message.
The code is universal in the sense that it does not depend on knowledge of the
structure of the channel $W^{BE}$. The only assumption  in the coding theorem
is that Alice and Bob have prior knowledge of the input distribution $\bp$ on
the set ${\cal{X}}$, and of bounds on the corresponding Holevo quantities for
the channels $W^B$ and $W^E$.

As a first step towards proving a universal private coding theorem, we derive
an alternative proof of a universal coding theorem for a memoryless
$c\rightarrow q$ channel (see  Theorem \ref{U_packing} of Section
\ref{packing}). Our coding theorem for the $c\rightarrow q$ channel $W^B$, only
requires prior knowledge of the probability distribution on the input of the
channel. It establishes the existence of a universal code using which Alice and
Bob can achieve reliable information transmission through the $c\rightarrow q$
channel at any rate less than the corresponding Holevo quantity. Our proof
employs a "type decomposition" and the random coding technique, but unlike
Hayashi's proof \cite{Hayashi09CMP}, it does not employ irreducible
representations and Schur-Weyl duality. However, the decoding POVM in our
universal code is analogous to his, which results in some of the steps of our
proof being similar. Moreover, like his result,
our theorem can be essentially viewed as a universal version of a cornerstone
of information theory, namely, the packing lemma \cite{CK81,HDW05EAC}.

The universal packing lemma ensures that Bob correctly infers Alice's messages
in the asymptotic limit. In addition, we require that these messages cannot be
inferred by the eavesdropper, Eve. This obliteration of information transmitted
over the channel $W^E$, induced between Alice and Eve, is established by
employing the  so-called covering lemma\cite{Winter01a,HLB08SKP},
which we prove explicitly below. We also establish that the covering lemma is
universal because it does not require Alice to have any specific knowledge of
the channel  $W^E$. It only depends on the Holevo quantity corresponding to the
input distribution of the channel. The universal covering lemma, when combined
with the universal packing, yields our main result, namely, the universal
private coding theorem.

%?? We then give an explicit proof of the so-called covering ???
%lemma \cite{Winter01a,HLB08SKP}, which when combined with the universal
%packing, yields the universal private coding theorem.??

Universal coding theorems have been established for other
information-processing tasks, for example, data compression \cite{JHHH98PRL,
JP03PRSLA, Hayashi:02b} and entanglement concentration
\cite{Hayashi:02a,BCG09}. Jozsa {\em{et al}} \cite{JHHH98PRL} introduced a universal data compression
scheme which did not require any knowledge about the information
source, other than an upper bound on its von Neumann entropy. The compression
scheme was proved to achieve a rate equal to this upper bound. Hence, our first
result, Theorem \ref{U_packing} of Section \ref{packing}, can be viewed as the
$c\rightarrow q$ channel counterpart of this result. Similarly, our second (and
main) result, Theorem \ref{U_private} of Section \ref{private}, is in a way the
$c\rightarrow qq$ channel counterpart of this same result, under the additional
requirement of privacy.

Note that the work by Jozsa {\em{et al}} \cite{JHHH98PRL} was followed by fully universal quantum data
compression schemes, presented first by Hayashi and Matsumoto
\cite{Hayashi:02b} and then by Jozsa and Presnell \cite{JP03PRSLA}, in which
the von Neumann entropy of the quantum information source was not known apriori
but was instead estimated.

In Section \ref{prelim} we introduce the relevant notations and definitions. In
Section \ref{packing} we prove a universal coding theorem for a $c\rightarrow
q$ channel. The Universal covering lemma is proved in Section \ref{covering}.
In Section \ref{private}, the results of the previous sections are combined to
prove our main result, namely a universal private coding theorem. We conclude
in Section~\ref{conclusion}.

\section{Notations and Definitions}\label{prelim}
Let ${\cal B}(\sH)$ denote the algebra of linear operators acting on a
finite--dimensional Hilbert space $\sH$ and let $\states(\sH)$ denote the set
of positive operators of unit trace (states) acting on $\sH$. We denote the
identity operator in ${\cB}(\sH)$ by $\iden$. For a state
$\rho\in\states(\sH)$, the von Neumann entropy is defined as
$S(\rho):=-\tr\rho\log\rho$. Further, for a state $\rho$ and a positive
operator $\sigma$ such that ${\rm{supp }} \rho \subseteq {\rm{supp }}\sigma$,
the quantum relative entropy is defined as $S(\rho||\sigma) = \tr \rho \log
\rho - \rho \log \sigma,$ whereas the relative R\'enyi entropy of order $\alpha
\in (0,1)$ is defined as
$$S_\alpha (\rho || \sigma) := \frac{1}{\alpha - 1} \log\bigl[ \tr(\rho^\alpha \sigma^{1-\alpha})\bigr].$$
Two entropic quantities, defined for any ensemble of states $\cE:= \{p_x,
\sigma_x \}_{x \in \cX}$, play a pivotal role in this paper. One is the Holevo
quantity, which is given in terms of the quantum relative entropy as follows:
\bea\label{chie} \chi(\cE) &=& \min_{\omega_Q} S(\sigma_{XQ} || \sigma_X
\otimes
\omega_Q)\nonumber\\
&=& S(\sigma_{XQ} || \sigma_X \otimes
\sigma_Q) \nonumber\\
&=& S\bigl(\sum_x p_x \sigma_x \bigr)-\sum_{x}p_x S(\sigma_x),\eea
where $\sigma_{XQ}$ is a classical-quantum state
$$\sigma_{XQ}:= \sum_x p_x |x\rangle \langle x| \otimes \sigma_x,$$
and $\sigma_X$, $\sigma_Q$ denote the corresponding reduced states.
The second identity in \reff{chie} follows from the fact \cite{fbnd}:
\be\label{sequal}
\min_{\omega_Q} S(\sigma_{XQ} || \sigma_X \otimes
\omega_Q) = S(\sigma_{XQ} || \sigma_X \otimes
\sigma_Q).
\ee

The other relevant entropic quantity is the $\alpha$-$\chi$ quantity,
which is defined for any $\alpha \in
(0,1)$ in terms of the relative R\'enyi entropy of order $\alpha$ as follows:
\bea\label{chia}
\chi_\alpha(\cE) &:=& \min_{\omega_Q} S_\alpha (\sigma_{XQ} || \sigma_X \otimes
\omega_Q)\nonumber\\
&=& \frac{\alpha}{\alpha - 1}\log{\rm Tr} \left[ \sum_{x\in{\cal X}}p_x
\sigma(x)^\alpha\right]^{\frac{1}{\alpha}},
\eea
(For a proof of the last identity, see e.g. \cite{KW}).

It is known that (see e.g. \cite{ohya})
$$\lim_{\alpha \nearrow 1} S_\alpha(\rho||\sigma) = S(\rho|| \sigma).$$
Moreover, it has been proved (see Lemma B.3 \cite{milan_hiai}) that
\be
\lim_{\alpha  \nearrow 1}\chi_\alpha(\cE) = \chi(\cE).
\label{chiequal}
\ee
Throughout this paper
we take the logarithm to base $2$ and restrict our considerations to finite-dimensional Hilbert spaces.

The trace distance between two operators $A$ and $B$ is given by
\begin{equation}\nonumber
  \N{A-B}_1 := \tr\bigl[\{A \ge B\}(A-B)\bigr] - \tr\bigl[\{A <
  B\}(A-B)\bigr],
\end{equation}
where $\{A\ge B\}$ denotes the projector onto the subspace where the operator
$(A-B)$ is non-negative. We make use of the following lemmas:

%\begin{lemma}[\cite{Bowen-Datta06}]\label{bowen}
%  For self-adjoint operators $A$ and $B$, and any positive
%  operator $0\le P\le\openone$,
%\begin{equation}\nonumber
%\Tr[P(A-B)]\le\Tr[\{A\ge B\}(A-B)].
%\end{equation}
%\end{lemma}

\begin{lemma}\cite{Datta-Renner09}
\label{lem2}
Given a state $\rho$ and a self-adjoint
operator $\omega$, for any real $\gamma$ we have
$$
\mathrm{Tr}\big[\{\rho \ge 2^{-\gamma}\omega \} \omega \bigr]
\leq 2^{\gamma}.
$$
\end{lemma}

\begin{lemma}[Gentle measurement lemma~\cite{AW99,ON02ISIT}]
\label{gm}
For a state $\rho\in\states(\sH)$ and operator $0\le
\Lambda\le\openone$, if $\Tr(\rho\Lambda) \ge 1 - \delta$, then
$$\N{\rho -   {\sqrt{\Lambda}}\rho{\sqrt{\Lambda}}}_1 \le {2\sqrt{\delta}}.$$
The same holds if $\rho$ is a subnormalized density operator.
\end{lemma}

\begin{lemma}[Operator Chernoff bound~\cite{Ahlswede-Winter02}] \label{OCB}
Let $\sigma_1,\cdots\sigma_N$ be independent and identically distributed random
variables with values in $\cB(\sH)$, which are bounded between 0 and the
identity operator $\openone$. Assume that the expectation value $\bbE
\sigma_m=\Omega\geq t\openone$ for some $0<t<1$. Then for every $0<\eps<1/2$
\begin{equation}
\Pr \left\{\frac{1}{N}\sum_{m=1}^N\sigma_m\not\in [1\pm\eps]\Omega\right\}\leq 2\dim \sH 2^{\left(-Nk{\eps^2t}\right)}
\end{equation}
where $k:=1/(2 (\ln 2)^2)$, and $[1\pm\eps]\theta=[(1-\eps)\theta;(1+\eps)\theta]$ is an interval in the
operator order: $[A;B]=\{\sigma\in\cB(\sH):A\leq \sigma\leq B\}$.
\end{lemma}

\begin{lemma}\cite{Hayashi09CMP}\label{lem_haya1}
For any operator $A\ge 0$ and $t\in(0,1)$, we have
$$ \max_{\sigma \in \states(\sH)} \tr \left( A \sigma^t\right)
=  \Bigl[ \tr \bigl(A^{\frac{1}{1-t}} \bigr)\Bigr]^{1-t}.$$
\end{lemma}

We state here a number of standard facts about types and typical sequences
\cite{CK81,CT91}, which we use in this paper.
%$\cX=\{1,2,\cdots, k\}$
Let $\cX$ denote a finite classical alphabet of size $|\cX|=k$, and let the
letters of the alphabet $\cX$ be ordered (e.g. lexicographically):
$$
\cX := \{\x_1, \x_2,\ldots, \x_k ; \x_1\le \x_2 \le \ldots \le \x_k \}.
$$
Let us denote by $N(\x_i|x^n)$ the number of occurrences of the symbol $\x_i
\in\cX$ in the sequence $x^n:=(x_1,\cdots,x_n) \in \cX^n$. We define the
{\em{type}} $t(x^n)$ of a sequence $x^n\in \cX^n$ as follows: $t(x^n):=\bq$,
where $\bq= (q_1, q_2, \ldots, q_k)$, denotes a probability vector of length
$k$ with elements
$$q_i=\frac{N(\x_i|x^n)}{n}.$$
Let $\cP^n_{\cX}$
denote the set of types in $\cX^n$. The size of
$\cP^n_\cX$ is bounded as follows \cite{CT91}:
\begin{equation}
|\cP^n_\cX| \leq (n+1)^k = 2^{n\zeta_n(k)}, \label{EQ_size_type}
\end{equation}
where, for any fixed integer $k$, we define
\begin{equation}\label{EQ_zeta}
\zeta_n(k):=\frac{k}{n}\log (n+1).
\ee
Note that $\zeta_n(k) \to 0$ as $n\to \infty$.
Define the set of sequences of type $\bq$ in $\cX^n$ by
$$\cT_\cX^n(\bq)=\{x^n \in \cX^n: t(x^n)=\bq\}.$$
For any type $\bq\in\cP^n_\cX$, we have \cite{CT91}:
\begin{equation}\label{EQ_size_type_sequence}
2^{n\left[ H(\bq)-\zeta_n(k)\right]} \leq |\cT_\cX^n(\bq)| \leq 2^{nH(\bq)}.
\end{equation}

Throughout this article, we denote the probability distribution on the set
$\cX$ by $\bp$, i.e., $\bp=(p_1,p_2,\cdots,p_k)$, where $p_i := p_{\x_i}$ for ${\x_i
\in \cX}$. The Shannon entropy of $\bp$ is defined as $H(\bp)=-\sum_{i=1}^k
p_i\log p_i$. For any $\delta>0$, define $\cP_{\bp,\delta}^n:=\{\bq\in\cP^n_\cX:
|q_i-p_i|\leq p_i\delta, \ \forall \x_i\in\cX\}.$
%{{Note that the Shannon
%entropies of the distributions $\bp$ and $\bq\in\cP_{\bp,\delta}^n$ are related
%as follows:
%\begin{equation}
%|H(\bq) - H(\bp)| \le \frac{\delta}{2}\log k + H(\frac{\delta}{2}) :=c(k,\delta)
%\label{fannes}
%\end{equation}
%where $c(k,\delta)\to 0$ and $\delta\to 0$. This can be seen as a special case of Fannes' inequality \cite{Fannes73}.
%}}
Define the set of $\delta$-typical sequences of length $n$ as
\begin{align}
\cT_{\bp,\delta}^n &= \bigcup_{\bq\in\cP_{\bp,\delta}^n}\cT_\cX^n(\bq)\nonumber \\
&=\left\{x^n\in\cX^n:\left|\frac{N(\x_i|x^n)}{n}-p_i\right|\leq
p_i\delta, \forall \x_i\in\cX\right\}. \label{EQ_typical_set}
\end{align}
For any $\eps,\delta>0$, some positive constant $c:=H(\bp)$ depending only on $\bp$, and sufficiently large
$n$, we have \cite{CT91}:
\begin{align}
Q_n:=\Pr\{X^n\in\cT_{\bp,\delta}^n\} &\geq 1-\eps \label{EQ_CT_1} \\
2^{-n[H(\bp)+c\delta]}\leq p^n_{x^n} &\leq 2^{-n[H(\bp)-c\delta]}\label{EQ_CT_2}, \forall x^n\in\cT_{\bp,\delta}^n \\
|\cT_{\bp,\delta}^n|&\leq 2^{n[H(\bp)+c\delta]},\label{EQ_CT_3}
\end{align}
where $p^n_{x^n}$ denotes the probability of the sequence $x^n$ and is given by
the product distribution $p^n_{x^n}:=\prod_{i=1}^n p_{x_i}$. We also use the
following bound \cite{CT91}: For any type $\bq\in\cP_{\bp,\delta}^n$,
\begin{equation}\label{EQ_CT_4}
|\cT_\cX^n(\bq)|\geq 2^{n\left[H(\bp)-\eta(\delta)\right]},
\end{equation}
where
%\be \label{eta}
%\eta = ....
%\ee
%Note that
$\eta(\delta)\to 0$ as $\delta\to 0$.

Consider a Hilbert space $\sH$, where $\dim\sH= d$. Let
$\cY=\{1,2\cdots,d\}$. It follows that $\sH^{\otimes
n}=\Span\{\ket{y^n}\equiv\ket{y_1}\otimes\cdots\otimes\ket{y_n}:\forall
y^n\in\cY^n\}.$ Let $\cK_\bq=\Span\{\ket{y^n}: y^n\in\cT_\cY^n(\bq)\}$, where
$\cT_\cY^n(\bq)$ is the collection of sequences of type $\bq$ in $\cY^n$. Then
$$\sH^{\otimes n}=\bigoplus_{\bq\in \cP^n_\cY}\cK_\bq,$$
where $\cP_\cY^n$ is the collection of all types in $\cY^n.$
Let $\tilde{\cK}_\bq=\Span\{U^{\otimes n}\ket{y^n}:\forall U\in {\rm{U}}(d), y^n\in\cT_\cY^n(\bq) \}$, where ${\rm U}(d)$ is the group of $d\times d$ unitary matrices. Note that $\tilde{\cK}_\bq$ is not associated with a preferred basis, unlike $\cK_\bq$.
Let $I_\bq\in\cB(\sH^{\otimes n})$ be the
projector onto $\cK_\bq$:
\begin{equation}\label{iq}
I_\bq =\sum_{y^n\in\cT_\cY^n(\bq)}\proj{y^n},
\end{equation}
and let $\tilde{I}_\bq\in\cB(\sH^{\otimes n})$ be the projector onto $\tilde{\cK}_\bq$. Since $\cK_\bq\subseteq\tilde{\cK}_\bq$, we also have $I_\bq\leq \tilde{I}_\bq$.
Define the maximally mixed state on $\tilde{\cK}_\bq$ to be:
\begin{equation}
\tau_\bq:=\frac{\tilde{I}_\bq}{|\tilde{\cK}_\bq|}, \label{EQ_tau_q}
\end{equation}
where $|\tilde{\cK}_\bq|$ is the dimension of the space $\tilde{\cK}_\bq$.
The following inequality holds:
\begin{equation}\label{eq_dim}
|\tilde{\cK}_\bq|\leq (n+1)^{d^2}|\cK_\bq|.
\end{equation}
For sake of completeness, we provide a proof of the above inequality in Appendix~\ref{app_dim}. Further define
\begin{align}
\tau_{n}&:= \frac{1}{|\cP^{n}_\cY|}\sum_{\bq\in \cP^{n}_\cY}\tau_{\bq}. \label{EQ_tau_n}
\end{align}
Note that $\tau_n$ does not depend on the choice of the initial basis $\{\ket{y^n}\}$. Consider a state $\sigma$ whose spectral decomposition is given by
\begin{equation} \sigma=\sum_{y=1}^d \lambda_y\proj{y}.\label{EQ_spectral}
\end{equation}
Then $S(\sigma)=H(\underline{\lambda})$, where
$\underline{\lambda}:=(\lambda_1,\cdots,\lambda_d)$.

\begin{lemma}\label{lemma1} For any state $\sigma\in\cB(\sH)$,
\begin{equation}
(n+1)^{(d^2+d)}\tau_n \geq\sigma^{\otimes n}. \label{EQ_sigma_size_bound}
\end{equation}
\end{lemma}
\begin{proof}
Assume that the spectral decomposition of $\sigma$ is given by (\ref{EQ_spectral}). Then \cite{CT91}
\begin{equation}\label{EQ_sigma_type}
\sigma^{\otimes n}=\sum_{\bq\in \cP^n_\cY}
2^{-n\left[D(\bq\|\underline{\lambda})+H(\bq)\right]} I_\bq,
\end{equation}
where $D(\bq\|\underline{\lambda}):=\sum_{i=1}^d q_i\log
\frac{q_i}{\lambda_i}$ denotes the relative entropy.  Since
\begin{align}
I_\bq \sigma^{\otimes n} I_\bq &=2^{-n\left[D(\bq\|\underline{\lambda})+H(\bq)\right]} I_\bq \nonumber\\
& \leq 2^{-nD(\bq\|\underline{\lambda})}\frac{I_\bq}{|\cK_\bq|} \nonumber \\
&\leq (n+1)^{d^2}\frac{\tilde{I}_\bq}{|\tilde{\cK}_\bq|}\nonumber\\
&=(n+1)^{d^2} \tau_\bq. \label{EQ_IqIq}
\end{align}
In the above, the first inequality follows from the following fact (similar to
(\ref{EQ_size_type_sequence})): $|\cK_\bq|=|\cT_{\cY}^n(\bq)|\leq 2^{nH(\bq)}$ for all
$\bq\in\cP_\cY^n$. The second inequality follows from the non-negativity
of $D(\bq\|\underline{\lambda})$, $I_\bq\leq \tilde{I}_\bq$, and (\ref{eq_dim}). The final equality follows from the definition (\ref{EQ_tau_q}) of $\tau_\bq$. Then
\begin{align*}
\sigma^{\otimes n} &= \sum_{\bq\in\cP_\cY^n} I_\bq \sigma^{\otimes n} I_\bq \nonumber \\
&\leq (n+1)^{d^2}\sum_{\bq\in\cP_\cY^n} \tau_\bq \\
&=(n+1)^{d^2}|\cP_\cY^n| \tau_n \nonumber\\
&\leq (n+1)^{(d^2+d)} \tau_n,
\end{align*}
where the first inequality follows from (\ref{EQ_IqIq}),
the next identity follows from (\ref{EQ_tau_n}), and the second
inequality follows from the fact (similar to (\ref{EQ_size_type})):
$|\cP_\cY^n|\leq (n+1)^d$.
\end{proof}

Let $x_o^n=(\x_1,\ldots,\x_1,\x_2,\ldots,\x_2,\ldots,\x_k,\ldots, \x_k)$ be the
ordered sequence in $\cT^n_\cX(\bq)$, where the number of $\x_i$ in $x_o^n$ is
$N(\x_i|x_o^n)=m_i=n q_i$. Define
\begin{align}
\omega_{x_o^n}&:=\tau_{m_1}\otimes\cdots\otimes\tau_{m_k},\label{EQ_omega_o}
\end{align}
where each $\tau_{m_i}$ is defined similarly to (\ref{EQ_tau_n}).

For any $x^n\in\cT_\cX^n(\bq)$, there exists a permutation $s\in S_n$ such that
$x^n=s x_o^n$. Let $U_s$ be the unitary representation of $s$ in $\sH^{\otimes
n}$. We can then define the following state that plays an important role in the
following sections:
\begin{equation}\label{EQ_omega_n}
\omega_{x^n}=U_s\omega_{x_o^n}U_s^\dagger.
\end{equation}

We define the $\delta$-typical projector $\Pi_{\sigma,\delta}^n$ of $\sigma^{\otimes n}$ to be
$$\Pi_{\sigma,\delta}^n=\sum_{y^n\in\cT_{\underline{\lambda},\delta}^n} \proj{y^n},$$
where $\sigma$ is defined in (\ref{EQ_spectral}) and
$\cT_{\underline{\lambda},\delta}^n$ is similarly defined as
$\cT_{\bp,\delta}^n$ in (\ref{EQ_typical_set}).

For any $\eps,\delta>0$, some positive constant $c:=H(\bp)$ depending only on $\bp$, and sufficiently large
$n$, we have \cite{NC00}:
\begin{align}
\tr \sigma^{\otimes n} \Pi_{\sigma,\delta}^n &\geq 1-\eps \label{top}\\
2^{-n\left[S(\sigma)+c\delta\right]}\Pi_{\sigma,\delta}^n &\leq \Pi_{\sigma,\delta}^n \sigma^{\otimes n}\Pi_{\sigma,\delta}^n \leq 2^{-n\left[S(\sigma)-c\delta\right]}\Pi_{\sigma,\delta}^n \\
\tr \Pi_{\sigma,\delta}^n &\leq 2^{n\left[S(\sigma)+c\delta\right]}. \label{bot}
\end{align}

\section{Universal Packing}\label{packing}
Consider a classical-quantum channel $W^B: x \to W^B(x)$ which maps the input
alphabets $\cX$ (with probability distribution $\bp$ on $\cX$) to the set of
states in $\states(\sH_B)$, (where ${\rm{dim }}{\cal{H}}_B = d_B$). Then using
this channel $n$ times gives a memoryless channel $W^{B^n}\equiv (W^B)^{\otimes
n}$ that maps $x^n\in\cX^n$ to a tensor product state in
$\states(\sH_B^{\otimes n})$: \be W^{B^n}(x^n):=W^B(x_1)\otimes\cdots\otimes
W^B(x_n). \label{prodw} \ee

%%%%%%%%% insertion %%%%%%%%%%%%%%

Suppose Alice wants to send classical messages in the set $\cM_n := \{1,2,
\ldots, M_n\}$ to Bob through the channel $W^{B^n}$. In order to do this she
needs to encode her messages appropriately before sending them through the
channel and Bob needs a decoder to decode the messages that he receives. The
encoding performed by Alice is a map $\varphi_n$ from the set of messages
$\cM_n$ to a set $\cA_n\subset\cX^n$. The decoding performed by Bob is a POVM
$\Upsilon^n:=\{\Upsilon_i\}_{i=1}^{M_n}$, where each POVM element $\Upsilon_i$
is an operator acting on $\sH_B^{\otimes n}$.

A ``c-q'' code $\cC_n(W^B)$ is given by the triple $\cC_n(W^B):=\{M_n,
\varphi_n,\Upsilon^n\}$, where $M_n$ denotes the size of the code (i.e., the
number of codewords). The \emph{average error probability} of the code
$\cC_n(W^B)$ is given by: \be p_e\left(\cC_n(W^B)\right):=\frac{1}{M_n}
\sum_{i=1}^{M_n}\tr
W^{B^n}(\varphi_n(i))\left(\iden-\Upsilon_i\right)\label{EQ_errop_p}.  \ee Let
$\cC(W^B):= \{\cC_n(W^B)\}_{n=1}^\infty$ denote a sequence of such c-q codes.
For such a sequence of codes, a real number $R$,
\begin{equation}\label{rate}
R:= \lim_{n \rightarrow \infty} \frac{1}{n} \log M_n
\end{equation} is called an achievable rate if
\begin{align*} p_e\left(\cC_n(W^B)\right) \rightarrow 0 \quad {\hbox{as}}\ \ \,\, n
\rightarrow \infty. \end{align*} We will refer to a sequence of codes
simply as a code when there is no possibility of ambiguity.

%%%%%%%%%% end insertion %%%%%555555

It has been shown that \cite{Hol98,SW97} for every classical-quantum channel
$W^B$, and any probability distribution $\bp$ on $\cX$, there exists a sequence
of c-q codes $\cC(W^B)$ with achievable rate
$$\chi(\bp,W^B):=S(\bar{W}^B)-\sum_{i=1}^k p_iS(W^B(\x_i)),$$ where
\begin{equation}
\bar{W}^B=\sum_{i=1}^k p_i W^B(\x_i).
\label{barb}
\end{equation}
Note that $\chi(\bp,W^B)$ is just the Holevo quantity of the output ensemble
$\{p_x, W^B(x)\}_{x \in \cX}$ of the classical-quantum channel. However, the
decoding POVM $\Upsilon^n$ of the code $\cC_n(W^B)$ constructed
requires knowledge of the channel $W^B$ \cite{Hol98,SW97}.

Hayashi \cite{Hayashi09CMP} proved that it is possible to construct a decoding POVM that is
independent of a given channel $W^B$. He then constructed a sequence of
universal c-q codes with achievable rate approaching $\chi(\underline{p},W^B)$
by combining Schur-duality and a classical universal code proposed by Csisz\'ar
and K\"orner. We use a decoding POVM similar to Hayashi's
and construct a different sequence of universal c-q codes. Unlike Hayashi's
methods of code construction, we employ ``type decomposition'' and the regular
random coding technique, though certain parts of the proof are similar to those
used by Hayashi \cite{Hayashi09CMP}.

%Our main result in this section is to show that one can construct a decoding
%POVM which does not require specific knowledge of the channel $W^B$ for their
%construction. Certain parts of the proof are similar to those used by Hayashi
%in \cite{Hayashi09CMP}, even though, unlike him, we do not employ irreducible
%representations.
\begin{theorem}\label{U_packing}
Given a probability distribution $\bp$ on the input alphabet $\cX$, of any
classical-quantum channel $W^B$, there exists a sequence, $\cC(W^B)$, of codes
which can achieve any rate $R <\chi(\bp,W^B)$, and are universal, in the sense
that their decoding POVMs only depend on $\chi(\bp,W^B)$ and not on specific
knowledge of the structure of the channel $W^B$.
\end{theorem}

\begin{proof}
We will employ the random coding technique to show the existence
of such a universal code $\cC$.

Let $M_n=|\cM_n|$. Let $\cA_n:=\{X_i\}_{i=1}^{M_n}$, where each $X_i$ is a
random variable chosen independently, according to
\begin{equation}\label{EQ_pro_cM}%
p'^n_{x^n}:=\Pr\{X_i=x^n\}=\begin{cases}
p^n_{x^n}/Q_n,& \text{if},\ x^n\in\cT_{\bp,\delta}^n  \\ 0&
\text{otherwise,}\end{cases}
\end{equation}
where $Q_n$ is defined in \reff{EQ_CT_1}. It is easy to verify that
\begin{equation}\label{prob_dis}
\|p-p'\|_1=\sum_{x^n} |p^n_{x^n}-p'^n_{x^n}|\leq 2\eps.
\end{equation}
The codeword $\varphi_n(i)$ for the $i$-th message is given by the realization
of the random variable $X_i$, taking values $x^n\in\cT_{\bp,\delta}^n$.

To construct a suitable POVM, we define the projector similar to Hayashi's
\begin{equation}\label{EQ_proj}
\Lambda_{X_i}:=\{\omega_{X_i}-2^{n\gamma_n}\tau_n\geq0\},
\end{equation}
where $\gamma_n$ is a real number to be determined below.

Define the POVM element $\Upsilon_{i}$:
\begin{equation}\label{POVM_packing}
\Upsilon_{i}:=\left(\sum_{j=1}^{M_n}\Lambda_{X_j}\right)^{-1/2} \Lambda_{X_i}\left(\sum_{j=1}^{M_n}\Lambda_{X_j}\right)^{-1/2},
\end{equation}
with $\Lambda_{X_i}$ being given by (\ref{EQ_proj}). Apparently, each
$\Lambda_{X_i}$ (and therefore each POVM element $\Upsilon_i$) does not depend
on full knowledge of the channel $W^B$. We then define the random universal
code $\cC_{\text{rc}}^n:=\{X_i,\Upsilon_i\}_{i=1}^{M_n}$.

Using the following operator inequality \cite{HN03}:
\begin{equation*}
\iden-\sqrt{S+T}^{-1}S\sqrt{S+T}^{-1}\leq 2(\iden-S)+4T,
\end{equation*}
where $0\leq S\leq \iden$ and $T\geq0$, the average error probability
$p_e(\cC_{\text{rc}}^n)$ in (\ref{EQ_errop_p}) can be bounded as follows:
\begin{equation}
p_e(\cC_{\text{rc}}^n)\leq\frac{2}{M_n} \sum_{i=1}^{M_n} \tr \left(\iden-\Lambda_{X_i}\right)W^{B^n}(X_i)
+\frac{4}{M_n}\sum_{i=1}^{M_n} \sum_{j=1, j\neq i}^{M_n}\tr
\Lambda_{X_j}W^{B^n}(X_i). \label{EQ_bounded_error_p}
\end{equation}
Taking the expectation of (\ref{EQ_bounded_error_p}),
with respect to the distribution (\ref{EQ_pro_cM}), we get
\begin{equation}\label{EQ_exp_error_p}
\bbE\left[p_e(\cC_{\text{rc}}^n)\right]\leq
\frac{2}{M_n} \sum_{i=1}^{M_n}
\bbE_i\left[\tr\left(\iden-\Lambda_{X_i}\right)W^{B^n}(X_i)\right] +\frac{4}{M_n}\sum_{i=1}^{M_n} \sum_{j=1\atop{j\neq i}}^{M_n} \bbE_{ij}\left[\tr \Lambda_{X_j}W^{B^n}(X_i)\right].
\end{equation}
We evaluate the sum in the first term of the above inequality in
Sec.~\ref{sec_sum1} and the inner sum in the second term in Sec.~\ref{sec_sum2}.
These results (in particular, (\ref{EQ_Pe_a}) and (\ref{EQ_Pe_b}) below) yield
the following upper bound on $\bbE\left[p_e(\cC_{\text{rc}}^n)\right]$:
%%% insert %%%
\noindent For any $0<t< 1$ we have that
\begin{equation}\label{e16}
\bbE\left[p_e(\cC_{\text{rc}}^n)\right] \le
2^{-nt\left[\chi_{1-t}-\gamma_n-\zeta_n(k(d_B^2+d_B))\right]} + 4(1-
\eps)^{-1}M_n 2^{-n(\gamma_n -
\zeta_n(d_B^2+d_B))}+2\eps
\end{equation}
where $\chi_{1-t} \equiv \chi_{1-t}(\{p_x, W^B(x)\})$ is the the $(1-t)$-$\chi$
quantity of the ensemble $\cE:=\{p_x, W^B(x)\}$, given by \reff{chia}, and
$\zeta_n(\cdot)$ is defined through \reff{EQ_zeta}. Denote
\begin{equation}
\eps_n:=2^{-nt\left[\chi_{1-t}-\gamma_n-\zeta_n(k(d_B^2+d_B))\right]} + 4(1-
\eps)^{-1}M_n 2^{-n(\gamma_n -
\zeta_n(d_B^2+d_B))}+2\eps.
\end{equation}

Our aim is to show that by choosing $\gamma_n$ appropriately, we can ensure
that $\eps_n \rightarrow 0$ as $n \rightarrow \infty$ for any $R <
\chi(\bp,W^B)$. This implies in particular, that if all that is known about the
classical-quantum channel is the probability distribution $\bp$ on $\cX$ and a
lower bound (say $\chi_0$) to the value of the corresponding Holevo quantity,
$\chi(\bp,W^B)$, then there exists a universal code $\cC(W^B)$, which can
achieve any rate $R < \chi_0$.

Let
\begin{equation}
M_n = 2^{n\left[R -  \zeta_n((k+1)(d^2_B+d_B))\right]}.
\label{enn}
\end{equation}
Note that this choice respects the definition \reff{rate} of a rate $R$, since
$\zeta_n((k+1)(d^2_B+d_B)) \rightarrow 0$ as $n \rightarrow \infty$. Further, for any
$t\in (0,1)$, let us choose
\begin{equation}
\gamma_n = R + r(t) - \zeta_n(k(d_B^2+d_B)),
\label{gamma_n}
\end{equation}
where we define \cite{Hayashi09CMP}
\begin{equation}
r(t):=\frac{t}{t+1}(\chi_{1-t} - R). \label{arrt}
\end{equation}
This choice of $\gamma_n$ reduces the second term on the RHS of \reff{e16} to
$4(1-\eps)^{-1}\times 2^{-nr(t)}$, and the first term on the RHS of \reff{e16} to
\begin{eqnarray}
&& 2^{-nt[\chi_{1-t} - R - r(t) + \zeta_n(k(d_B^2+d_B)) -  \zeta_n(k(d_B^2+d_B))]} \nonumber\\
&=& 2^{-nt[\chi_{1-t} - R - (\chi_{1-t} - R- \frac{r(t)}{t})]}\nonumber\\
&=& 2^{-nr(t)},
\end{eqnarray}
where we have used the fact that $r(t)$, defined by \reff{arrt}, is
equivalently expressed as $r(t) = \chi_{1-t} - R - {r(t)}/{t}.$ Hence we obtain
the following:
\begin{equation}
\bbE\left[p_e(\cC_{\text{rc}}^n)\right] \le (1+4(1-\eps)^{-1})\times 2^{-nr(t)}+2\eps.
\end{equation}
Since this holds for any $t\in (0,1)$, we have in particular that
\begin{equation}
\bbE\left[p_e(\cC_{\text{rc}}^n)\right] \le (1+4(1-\eps)^{-1})\times2^{-n \max_{t \in (0,1)} r(t)}+2\eps.
\end{equation}
To prove that $\bbE\left[p_e(\cC_{\text{rc}}^n)\right]$ vanishes asymptotically
for a suitable range of values of $R$, it suffices to prove that $\max_{t \in
(0,1)} r(t) >0$ for that range of values of $R$. From \reff{chiequal} we infer
that
\begin{equation} \lim_{t \searrow 0}\chi_{1-t} = \chi,
\end{equation}
where $\chi=\chi(\bp,W^B)$. Moreover, the convergence is from below, since $t
\mapsto \chi_{1-t}$ is monotonically decreasing. This ensures that for any
$R<\chi$, there exists a $t_R>0$ such that
$$ \chi_{1-t} - R >0, \quad \forall \,\, t < t_R.$$
This in turn implies that for any given $R < \chi$,
\begin{equation} \max_{t
\in (0,1)} r(t)
>0. \end{equation}
\end{proof}

%%%% end insert %%%%
Thus there must exist a sequence of codes $\cC_n$ of rate $R < \chi$ such that
$p_e(\cC_n) \rightarrow 0$ as $n \rightarrow \infty$. It follows that for any
$\eps >0$, and sufficiently large $n$,
$$p_e(\cC_n) = \frac{1}{M_n} \sum_{i=1}^{M_n}p_e(i)< \eps,$$
where $M_n$ is given by \reff{enn} and $p_e(i) := \tr
W^{B^n}(\varphi_n(i))(\bbI-\Upsilon_i)$ denotes the probability of error
corresponding to the $i^{th}$ message. %For this to be true, at least half the
%messages must correspond to a probability of error less than $2\eps$. So we
%construct a new code by deleting half the codewords (those with probability of
%error $>2\eps$, thus obtaining a new code with $M_n/2 = 2^{n(R -
%\zeta_n((k+1)(d^2_B+d_B)) - \frac{1}{n})}$ codewords, and with probability of error
%less than $2\eps$ for all messages. Obviously, this code also has a rate $R$,
%and the probability of error can be made arbitrarily small for {\em{all}}
%codewords as $n$ becomes large. Hence, any rate $R < \chi(\bp,W^B)$ is
%achievable.

\subsection{Evaluation of the first term in (\ref{EQ_exp_error_p})}\label{sec_sum1}
%%%%%%% begin insert %%%%%%%%%
In the section, we closely follow \cite{Hayashi09CMP}.

%%%%%%% end insert %%%%%%%%%%%
\noindent
\begin{lemma}
\be {\bbE}_i\Bigl[ {\rm Tr}(\bbI-\Lambda_{X_i})W^{B^n}(X_i)\Bigr] \le
2^{-nt\left[\chi_{1-t} -\gamma_n-\zeta_n(k(d_B^2+d_B))\right]}+2\eps \label{EQ_Pe_a}
\ee where $\chi_{1-t} \equiv \chi_{1-t}(\{p_x, W^B(x)\})$ is the
$(1-t)$-$\chi$ quantity of the ensemble $\cE:=\{p_x, W^B(x)\}$, defined through
\reff{chia}, and $\zeta_n(k(d_B^2+d_B))$ is defined through \reff{EQ_zeta}.
\end{lemma}

\begin{proof}

Consider the projector $\Lambda_{x^n}$ defined in \reff{EQ_proj}, where
$\gamma_n \in \bbR$. Since $[\omega_{x^n},\tau_n] =0$, and
this in turn can be shown to imply that for any $t\in (0,1)$ \cite{Hayashi09CMP}:
\begin{equation}\label{a}
(\bbI-\Lambda_{x^n}) \le \omega_{x^n}^{-t} 2^{nt\gamma_n} \tau_n^t.
\end{equation}
For sake of completeness, the commutativity
of the operators $\omega_{x^n}$ and $\tau_n$ is proved in Appendix~\ref{commu}. Since $W^{B^n}(x_o^n)=(W^B(\x_1))^{\otimes m_1}\otimes \cdots\otimes(W^B(\x_k))^{\otimes m_k}$, by direct application of Lemma \ref{lemma1}, we obtain
\begin{eqnarray}\label{b}
W^{B^n}(x^n)&=& U_s W^{B^n}(x_o^n) U_s^\dagger \nonumber\\
&\le & (n+ 1)^{k(d_B^2+d_B)} U_s\left(\tau_{m_1} \otimes \ldots \otimes
\tau_{m_k}\right) U_s^\dagger\nonumber\\
&=&  (n+1)^{k(d_B^2+d_B)}U_s \omega_{x_o^n} U_s^\dagger \nonumber \\
&= &  (n+1)^{k(d_B^2+d_B)} \omega_{x^n}. \label{c}
\end{eqnarray}
This yields, for any $t\in (0,1)$:
\begin{eqnarray}
\omega_{x^n}^{-t} &\le &  (n    + 1)^{tk(d_B^2+d_B)} \left(W^{B^n}(x^n)\right)^{-t}\nonumber\\
&=& 2^{nt\zeta_n(k(d_B^2+d_B))}  \left(W^{B^n}(x^n)\right)^{-t}, \label{cc1}
\end{eqnarray}
and
\begin{equation}\label{d}
W^{B^n}(x^n)\omega_{x^n}^{-t} \le 2^{nt\zeta_n(k(d_B^2+d_B))} \left(W^{B^n}(x^n)\right)^{1-t}.
\end{equation}
Finally,
\begin{eqnarray}
\mathbb{E}_i \Bigl[{\rm Tr}(\bbI-\Lambda_{X_i})W^{B^n}(X_i)\Bigr]
&=& \sum_{x^n\in\cT_{\bp,\delta}^n}p'^n_{x^n} \left[{\rm Tr}(\bbI-\Lambda_{x^n})W^{B^n}(x^n)\right] \nonumber\\
&\leq& \sum_{x^n\in\cX^n}p^n_{x^n} \left[{\rm Tr}(\bbI-\Lambda_{x^n})W^{B^n}(x^n)\right] +2\eps \nonumber\\
&\leq& 2^{nt\gamma_n}  \tr \left(\sum_{x^n\in{\cal X}^n}p^n_{x^n}\omega_{x^n}^{-t}   W^{B^n}(x^n)\right)\tau_n^t +2\eps \nonumber\\
&\le &  2^{nt\zeta_n(k(d_B^2+d_B))} 2^{nt\gamma_n} \left[{\rm Tr}  \left(\sum_{x^n\in{\cal X}^n}p^n_{x^n} \left(W^{B^n}(x^n)\right)^{1-t}\right)\tau_n^t\right] +2\eps \nonumber\\
&\le & 2^{nt\zeta_n(k(d_B^2+d_B))}2^{nt\gamma_n} \max_{\sigma}\Bigl[{\rm Tr}  \bigl(\sum_{x\in{\cal X}}p_x \left(W^{B}(x)\right)^{1-t}\bigr)^{\otimes n}\sigma^t\Bigr]+2\eps . \nonumber
\end{eqnarray}
The first equality follows from evaluating the expectation. The first
inequality follows from \reff{prob_dis}. The second inequality follows from
\reff{a}. The third inequality follows from \reff{d}.

%the above inequality uses Hayashi's Eq (13),
Applying Lemma \ref{lem_haya1} to the above equation, we get
\begin{eqnarray} &=&
2^{nt[\zeta_n(k(d_B^2+d_B))+\gamma_n]} \left\{{\rm Tr} \left[ \left(\sum_{x\in{\cal
X}}p_x \left(W^{B}(x)\right)^{1-t}\right)^{\otimes n}\right]^{\frac{1}{1-t}}\right\}^{1-t}+2\eps
\nonumber\\
&=& 2^{nt[\zeta_n(k(d_B^2+d_B))+\gamma_n]}\left\{{\rm Tr} \left[ \sum_{x\in{\cal X}}p_x \left(W^{B}(x)\right)^{1-t}\right]^{\frac{1}{1-t}}\right\}^{n(1-t)}+2\eps \nonumber\\
&=& 2^{-nt\left[\chi_{1-t} -\gamma_n-\zeta_n(k(d_B^2+d_B))\right]}+2\eps  \nonumber
\end{eqnarray}
where
$\chi_{1-t} \equiv  \chi_{1-t}(\{p_x, W^B(x)\})$.
\end{proof}

\subsection{Evaluation of the second term in (\ref{EQ_exp_error_p})} \label{sec_sum2}%
We have, for given $i,j \in \{1,2,\ldots, M_n\}$ and $i \ne j$,
\begin{eqnarray}
\bbE_{i,j}\left[\tr\Lambda_{X_j}W^{B^n}(X_i)\right]&=& \bbE_{j}\left[\tr\Lambda_{X_j} \bbE_{i}\left[W^{B^n}(X_i)\right]\right] \nonumber\\
&= & \bbE_{j}\left[\tr\left(\Lambda_{X_j}\frac{1}{Q_n}\sum_{x^n\in\cT_{\bp,\delta}^n}p^n_{x^n} W^{B^n}(x^n)\right)\right]\nonumber\\
&\leq&(1-\eps)^{-1} \bbE_{j}\tr\left(\Lambda_{X_j}\left(\bar{W}^B\right)^{\otimes n}\right)\nonumber\\
&\le &(1-\eps)^{-1} (n+1)^{(d_B^2+d_B)} \bbE_{j} \tr \bigl[\Lambda_{X_j} \tau_n\bigr] \nonumber \\
&\le &(1-\eps)^{-1} 2^{n\zeta_n(d_B^2+d_B)} 2^{-n\gamma_n}\nonumber \\
&= & (1-\eps)^{-1}2^{-n \left[\gamma_n-\zeta_n(d_B^2+d_B)\right]}. \label{EQ_Pe_b}
\end{eqnarray}
In the above, the first inequality follows from that fact $Q_n\geq 1-\eps$ and
\bea
\left(\bar{W}^B\right)^{\otimes n} &=& \sum_{x^n\in\cX^n}
p^n_{x^n}W^{B^n}(x^n) \nonumber\\
&\ge & \sum_{x^n\in\cT_{\bp,\delta}^n} p^n_{x^n}W^{B^n}(x^n).
\eea
The second inequality follows from Lemma \ref{lemma1}. The third inequality
follows from Lemma \ref{lem2}. $\zeta_n(d_B^2+d_B)$ in the last equality is defined
in \reff{EQ_zeta}.

%%%%%%%% begin insert %%%%%%%

\section{Universal Covering}\label{covering}

Consider the probability distribution $\bp=(p_1,\cdots,p_k)$ on $\cX$ and a
classical-quantum channel $W^E:x\in\cX\to W^E(x)\in\states(\sH_E)$, with $d_E =
{\rm{dim}} \, \sH_E$. Then using this channel $n$ times gives a memoryless
channel $(W^E)^{\otimes n}:=W^{E^n}$ that maps a sequence $x^n\in\cX^n$ with
probability $p^n_{x^n}:=p_{x_1}\cdots p_{x_n}$ to a product state in
$\states(\sH_E^{\otimes n})$:
$$W^{E^n}(x^n)=W^E(x_1)\otimes\cdots\otimes W^E(x_n).$$
Let $\bar{W}^E=\sum_{i=1}^k p_iW^E(\x_i).$ Then
$$\bar{W}^{E^n}:=(\bar{W}^E)^{\otimes
n}=\sum_{x^n\in\cX^n}p^n_{x^n}W^{E^n}(x^n).$$

Consider a subset $\cS\subset\cX^n$, and define the ``obfuscation error''
\begin{equation}\label{EQ_Delta}
\Delta(\cS):=\left\|\frac{1}{|\cS|}\sum_{x^n\in\cS}W^{E^n}(x^n)-\bar{W}^{E^n}\right\|_1.
\end{equation}
%%%%%%%%% insert %%%%%%%%%%%%
We are interested in finding the smallest ``covering'' subset $\cS\subset\cX^n$
for which $\Delta(\cS)\to 0$ as $n \to \infty$. We discover that for any given
probability distribution $\bp$, the covering set is universal in the sense that
its size $|\cS|$ depends only on the value of $\chi(\bp,W^E)$, and not on any
specific knowledge of the channel $W^E$ itself. Here we provide an explicit
proof of the fact that  $\Delta(\cS)$ can be made arbitrarily small for any
randomly picked subset $\cS\subset\cX^n$ as long as $n$ is sufficiently large
and $\frac{1}{n}\log|\cS|>\chi(\bp,W^E)$ \cite{HLB08SKP}, a result which follows
from the so-called ``covering lemma''  \cite{Winter01a,LD06CS,DHW06}. More
precisely, given any upper bound (say $\chi_1$) on the Holevo quantity
$\chi(\bp,W^E)$, there is a subset $\cS$ of the $\delta$-typical set,
$\cT_{\bp,\delta}^n$ (defined through (\ref{EQ_typical_set})) for any $\delta
>0$, for which $\Delta(\cS)\to 0$ as $n \to \infty$ provided $|\cS|>\chi_1$.

%%%%%%%%% end insert %%%%%%%%%%%%%%

%The main result in this section is to prove that there is a {\em{universal}}
%covering set $\cS$, i.e., a covering set whose construction does not rely
%on specific knowledge of the channel $W^E$, but only the value of $\chi(\bp,W^E)$.
%%We provide a new proof for the covering set that does not rely on
%%the knowledge of a specific channel $W^E$.
%The size of the universal covering set is equal to the size of original channel-dependant covering set, obtained in \cite{LD06CS,HLB08SKP}, asymptotically. More precisely, we show that
%for any $\delta >0$, there is a subset $\cS$ of $\cT_{\bp,\delta}^n$ (defined through (\ref{EQ_typical_set}), for which
%$\Delta(\cS)\to 0$ as $n \to \infty$ provided $|\cS|>\chi(\bp,W^E)$.
%%
\begin{theorem}\label{U_covering}
For any $\eps, \delta >0$, let $L_n=2^{n\left[\chi_1+2c\delta\right]}$, where
$\chi_1$ is a given upper bound on the Holevo quantity $\chi(\bp,W^E)$. Then
the set $\cA_n=\{X_i\}_{i=1}^{L_n}$, where each $X_i$ is a random variable
chosen independently according to
\begin{equation}\label{EQ_pro_M2}%
p'^n_{x^n}:=\Pr\{X_i=x^n\}=\begin{cases}
p^n_{x^n}/Q_n,& \text{if}\ x^n\in\cT_{\bp,\delta}^n,  \\ 0,&
\text{otherwise}\end{cases}
\end{equation}
($Q_n$ is given in \reff{EQ_CT_1}), satisfies
\begin{equation}
\Pr\Bigl\{\Delta(\cA_n)\geq \eps+4\sqrt{k\eps}+8\sqrt{3\eps+2\sqrt{k\eps}}\Bigr\} \leq \eps_n,
\end{equation}
for a positive constant $\eps_n$ (given in (\ref{EQ_epn_n})) such that $\eps_n \rightarrow 0$
as $n \rightarrow \infty$.
\end{theorem}

\begin{proof}
Recall that for any sequence $x^n \in \cX^n$, there exists a permutation $s\in
S_n$ such that $x^n=s x_o^n$, where $x_o^n$ denotes the ordered sequence
corresponding to $x^n$. Let $U_s$ be the unitary representation of $s$ in
$\sH_E^{\otimes n}$, and let us define
$$\Pi_{W^E,\delta}^n(x^n) = U_s\Pi_{W^E,\delta}^n(x_o^n)U_s^\dagger,$$
where $\Pi_{W^E,\delta}^n(x_o^n)$ is the $\delta$-conditional typical projector
for the state $W^{E^n}(x_o^n)$:
$$\Pi_{W^E,\delta}^n(x_o^n) \equiv \Pi_{W^E(\x_1),\delta}^{m_1}\otimes\cdots\otimes\Pi_{W^E(\x_k),\delta}^{m_k},$$
and $\Pi_{W^E(\x_i),\delta}^{m_i}$ is the $\delta$-typical projector for the
state $W^{E^{m_i}}(\x_i)$ such that
$$\tr \Pi_{W^E(\x_i),\delta}^{m_i} W^{E^{m_i}}(\x_i)\geq 1-\eps.$$
Let $$\sigma(x^n):=\Pi_{W^E,\delta}^n(x^n)
W^{E^n}(x^n)\Pi_{W^E,\delta}^n(x^n).$$
Then
\begin{align}
\tr\sigma(x^n)&=\tr \Pi_{W^E,\delta}^n(x^n) W^{E^n}(x^n) \nonumber \\
&= \tr \Pi_{W^E,\delta}^n(x_o^n) W^{E^n}(x_o^n) \nonumber \\
&= \prod_{i=1}^k \tr\Pi_{W^E(\x_i),\delta}^{m_i} W^{E^{m_i}}(\x_i)\nonumber \\
&\geq 1-k\eps. \label{EQ_C_tmp11}
\end{align}
Applying the gentle measurement lemma, Lemma \ref{gm}, to (\ref{EQ_C_tmp11}) gives
\begin{equation}\label{EQ_C_tmp2}
\|\sigma(x^n)-W^{E^n}(x^n)\|_1\leq 2\sqrt{k\eps}.
\end{equation}
Define the $\delta$-typical projector $\Pi_{\bar{W}^E,\delta}^n$ for the average state $\bar{W}^{E^n}$ such that
\begin{equation}\label{EQ_C_tmp22}
\tr \Pi_{\bar{W}^E,\delta}^n \bar{W}^{E^n} \geq 1-\eps.
\end{equation}
Define
$$\phi(x^n)=\Pi_{\bar{W}^E,\delta}^n \sigma(x^n)\Pi_{\bar{W}^E,\delta}^n.$$
Then
\begin{align}
\tr\phi(x^n)&=\tr\sigma(x^n)\Pi_{\bar{W}^E,\delta}^n \nonumber\\
&\geq \tr \Pi_{\bar{W}^E,\delta}^nW^{E^n}(x^n) -\|\sigma(x^n)-W^{E^n}(x^n)\|_1\nonumber \\
&\geq 1-\eps-2\sqrt{k\eps},  \label{EQ_C_tmp3}
\end{align}
where the first inequality follows from the fact that for any operator
$\omega$, $\|\omega\|_1=\max_{-\iden\leq \Pi\leq \iden}\tr \Pi\omega$, and the
second inequality follows from (\ref{EQ_C_tmp2}) and (\ref{EQ_C_tmp22}).
Applying the gentle measurement lemma to (\ref{EQ_C_tmp3}) gives
\begin{equation}\label{EQ_C_tmp8}
\|\phi(x^n)-\sigma(x^n)\|_1\leq 2\sqrt{\eps +2\sqrt{k\eps}}.
\end{equation}
Using the triangle inequality, and the bounds (\ref{EQ_C_tmp2}) and (\ref{EQ_C_tmp8}), yields
\begin{equation}\label{EQ_C_tmp7}
\|\phi(x^n)-W^{E^n}(x^n)\|_1\leq 2\sqrt{k\eps}+2\sqrt{\eps +2\sqrt{k\eps}}.
\end{equation}
Let
\begin{align*}
\bar{\phi}&:=\bbE \phi(X_i)=\sum_{x^n\in\cT_{\bp,\delta}^n} p'^n_{x^n}\phi(x^n).
\end{align*} Then (\ref{EQ_C_tmp3}) implies that
\begin{equation}
\tr\bar{\phi} \ge 1 - \eps-2\sqrt{k\eps} . \label{b1}
\end{equation}
Define
\begin{equation}\label{EQ_proj_hat}
{\Pi}:=\left\{\bar{\phi}- \eps2^{-n[S(\bar{W}^E)+c\delta]} \Pi_{\bar{W}^E,\delta}^n\geq 0\right\}.
\end{equation}
Note that $[\bar{\phi}, \Pi_{\bar{W}^E,\delta}^n] = 0$.

Let $\bar{\phi}':={\Pi}\bar{\phi}{\Pi}$. Then
\begin{align}
\tr\bar{\phi}' &\geq \tr\bar{\phi} - \eps \nonumber\\
&\geq 1-2\eps-2\sqrt{k\eps}, \label{EQ_C_tmp4}
\end{align}
where the first inequality follows from the fact that the support of
$\Pi_{\bar{W}^E,\delta}^n$ has dimension less than
$2^{n[S(\bar{W}^E)+c\delta]}$, so the eigenvalues smaller than
$\eps2^{-n[S(\bar{W}^E)+c\delta]}$ contribute at most $\eps$ to
$\tr\bar{\phi}$, while the second inequality follows from (\ref{b1}). The
following operator inequality holds:
\begin{align}
\bar{\phi}'&\geq
\eps2^{-n[S(\bar{W}^E)+c\delta]}\Pi_{\bar{W}^E,\delta}^n \label{b4}
%&\geq \eps2^{-n[S(\bar{W}^E)+\zeta_n(d_E)+3c\delta]}\Pi_{\widehat{\bar{W}}^E,3\delta}^n,
\end{align}
where the inequality follows from the definition (\ref{EQ_proj_hat}).%,

For any $ x^n\in\cT_{\bp,\delta}^n$, let $\theta(x^n):={\Pi}{\phi}(x^n){\Pi}$.
Note that these operators lie in the subspace of the Hilbert space
${\sH}_E^{\otimes n}$ onto which the operator $\Pi_{\bar{W}^E,\delta}^n$
project. Consider the operator ensemble
$\{p'^n_{x^n},\psi(x^n)\}_{x^n\in\cT_{\bp,\delta}^n}$, where
\begin{equation}
\psi(x^n):=2^{n\left[\sum_{i=1}^k p_i S(W^E(\x_i))-c\delta\right]}\theta(x^n).
\label{varphi}
\end{equation}
Note that
\begin{equation}
\bbE \psi(X_i) = 2^{n\left[\sum_{i=1}^k p_i
S(W^E(\x_i))-c\delta\right]}\bbE \theta(X_i), \label{b2}
\end{equation}
and
\begin{equation} \bbE \theta(X_i) ={\Pi}\bar{\phi}{\Pi}=
\bar{\phi}'.\label{b3}
\end{equation}
By (\ref{b2}), (\ref{b3}) and (\ref{b4}), we have
\begin{equation}
\bbE \psi(X_i) \ge
\eps2^{-n\left[\chi(\bp,W^E)+2c\delta\right]}\Pi_{\bar{W}^E,\delta}^n.
\end{equation}
Hence, the operator ensemble
$\{p'^n_{x^n},\psi(x^n)\}_{x^n\in\cT_{\bp,\delta}^n}$ satisfies the condition
of Lemma~\ref{OCB} with
\begin{equation}
t=\eps 2^{-n\left[\chi(\bp,W^E)+2c\delta\right]}.
\end{equation}

Applying Lemma~\ref{OCB}, we have
\begin{equation}\label{EQ_cover_bound}
\Pr\left\{\frac{1}{L_n}\sum_{i=1}^{L_n} \theta(X_i)\not\in
[1\pm\eps]\bar{\phi}'\right\} \leq 2 \tr(\Pi_{\bar{W}^E,\delta}^n)
2^{-L_n k {\eps^3 2^{-n[\chi(\bp,W^E)+2c\delta]}}}:= \eps_n^\prime
\end{equation}
where $k =1/(2 (\ln 2)^2)$, $L_n=2^{n[\chi_1+2c\delta]}$ and
\begin{align}\label{EQ_epn_n}
\eps_n^\prime&:=2\times 2^{-k\eps^3 2^{n[\chi_1-\chi(\bp,W^E)]}+n[S(\bar{W}^E)+c\delta]}\to 0
\end{align} as $n\to\infty.$

To prove Theorem \ref{U_covering}, the above inequality (\ref{EQ_cover_bound})
needs to be translated into a statement about the operators $W^{E^n}(x^n)$. To
do this, assume that for some set $\cA_n \subset \cT_{\bp,\delta}^n$ such that
for any $\eps >0$ and for $n$ large enough,
\begin{equation*} \frac{1}{|\cA_n|}\sum_{x^n \in \cA_n}
\theta(x^n)\in [1\pm\eps]\bar{\phi}'.
\end{equation*}
Equivalently, \be
\left\|\frac{1}{|\cA_n|}\sum_{x^n\in\cA_n}\theta(x^n)-\bar{\phi}'\right\|_1\leq\eps.
\label{c1}
\end{equation}
We shall show that \be\label{obfus} \Delta(\cA_n)\leq
\eps+4\sqrt{k\eps}+8\sqrt{3\eps+2\sqrt{k\eps}}. \ee From (\ref{c1}) and
(\ref{EQ_C_tmp4}), we obtain that
\begin{equation}\label{EQ_tmp6}
\tr\left(\frac{1}{|\cA_n|}\sum_{x^n\in\cA_n}\theta(x^n)\right) \geq 1-3\eps-2\sqrt{k\eps}.
\end{equation}
Since $\theta(x^n)=\Pi\phi(x^n)\Pi$, applying the gentle measurement lemma to
(\ref{EQ_tmp6}) gives
\begin{equation}\label{EQ_C_tmp5}
\left\|\frac{1}{|\cA_n|}\sum_{x^n\in\cA_n}\phi(x^n)-\frac{1}{|\cA_n|}\sum_{x^n\in\cA_n}\theta(x^n)\right\|_1\leq 2\sqrt{3\eps+2\sqrt{k\eps}}.
\end{equation}
Likewise, applying the gentle measurement lemma to (\ref{EQ_C_tmp4}) gives
\begin{equation}\label{EQ_C_tmp6}
\|\bar{\phi}'-\bar{\phi}\|_1\leq 2\sqrt{2\eps+2\sqrt{k\eps}}.
\end{equation}
We also have
\begin{align}
&\left\|\frac{1}{|\cA_n|}\sum_{x^n\in\cA_n}\phi(x^n)-\frac{1}{|\cA_n|}\sum_{x^n\in\cA_n}W^{E^n}(x^n)\right\|_1 \nonumber \\
&\leq \frac{1}{|\cA_n|}\sum_{x^n\in\cA_n} \left\|\phi(x^n)-W^{E^n}(x^n)\right\|_1 \nonumber\\
&\leq 2\sqrt{k\eps}+2\sqrt{\eps+2\sqrt{k\eps}},\label{EQ_C_tmp9}
\end{align}
where the second inequality follows from (\ref{EQ_C_tmp7}). We can analogously
obtain
\begin{align}
\|\bar{\phi}-\bar{W}^{E^n}\|_1&\leq 2\sqrt{k\eps}+2\sqrt{\eps+2\sqrt{k\eps}}. \label{EQ_C_tmp10}
\end{align}
By applying the triangle inequality and combining (\ref{EQ_C_tmp9}),
(\ref{EQ_C_tmp5}), (\ref{c1}), (\ref{EQ_C_tmp6}) and (\ref{EQ_C_tmp10}) we
obtain the desired bound (\ref{obfus}) for $\Delta(\cA_n)$. The statement of
the theorem follows immediately from (\ref{EQ_cover_bound}).

\end{proof}

%%%%%%%% end insert %%%%%%%%%%%
%%%% begin insert %%%%%%%%
\section{Universal Private Coding}\label{private}

Consider a $c\rightarrow qq$ channel $W^{BE}:x\to W^{BE}(x)$ that maps the
input alphabet $\cX$ (with the probability distribution $\bp$ on $\cX$) to the
set of densities in $\states(\sH_B\otimes \sH_E)$. Such a channel induces a
classical-quantum channel $W^B:x\to W^{B}(x)$, where $W^B(x)=\tr_E W^{BE}(x)$,
from the sender Alice to the receiver Bob. Meanwhile, it also induces a
classical-quantum channel $W^E:x\to W^E(x)$, where $W^E(x)=\tr_B W^{BE}(x)$,
from the sender Alice to an eavesdropper Eve.

The communication task is for Alice to send a private classical message in the
set $\cJ_n:= \{1,2,\cdots, J_n\}$ reliably to the receiver Bob through $n$ uses
of the $c\rightarrow qq$ channel, i.e., through $W^{B^nE^n}:=(W^{BE})^{\otimes
n}$, such that Eve cannot obtain any information about the message sent by
Alice. In order for them to achieve this goal, Alice and Bob require an encoder
and decoder, respectively. The encoding performed by Alice is a map $\varphi_n$
that maps a classical message $i\in\cJ_n$ to an arbitrary element
$x^n\in\cS_i\subset\cX^n$, where each ``covering set'' $\cS_i$ is disjoint and
has the same size. The decoding performed by Bob is a POVM
$\Upsilon^n:=\{\Upsilon_i\}_{i\in \cJ_n}$, where each POVM element is an
operator acting on $\sH_B^{\otimes n}$. For any arbitrary $\eps>0$, we can
formally define an $(n,\eps)$
``private'' code $\cC_n(W^{BE}):=\{J_n,\varphi_n,\Upsilon^n\}$ for the channel
$W^{BE}$ by
\begin{enumerate}
\item Alice's encoding $\varphi_n:i\to x^n\in\cS_i\subset\cX^n;$
\item Bob's decoding POVM $\Upsilon^n:\states(\sH_B^{\otimes n})\to \cJ_n$,
\end{enumerate}
such that for $n$ large enough, the following conditions hold:
\begin{itemize}
\item the average error probability of $\cC_n(W^{BE})$:
\begin{equation}
p_e(\cC_n(W^{BE})):=\frac{1}{J_n}\sum_{i=1}^{J_n} \tr \Bigl[W^{B^n}
(\varphi_n(i))(\bbI-\Upsilon_i)\Bigr]\leq\eps.\label{EQ_cond_p1}
\end{equation}
\item Eve cannot obtain any information on the classical message $i$ by
    measuring the state $W^{E^n}(x^n):= \tr_{B^n} W^{B^nE^n}(x^n)$ that she
    has access to, i.e., $\forall \,i \in \{1,2,\ldots, J_n\}$:
\begin{equation}
\Delta(\cS_i):=\left\|\frac{1}{|\cS_i|}\sum_{x^n\in\cS_i}W^{E^n}(x^n)-\bar{W}^{E^n}\right\|_1\leq\eps,\label{EQ_cond_p2}
\end{equation}
where $\bar{W}^{E^n}:= (\bar{W}^E)^{\otimes n}$ and $\bar{W}^E=\sum_{x \in \cX} p_{x} W^E(x)$.
\end{itemize}
\smallskip

Let $\cC(W^{BE}):= \{\cC_n(W^{BE})\}_{n=1}^\infty$ denote a sequence of such
private codes. For such a sequence of codes, a real number $R$,
\begin{equation}
R:=\lim_{n \rightarrow \infty} \frac{1}{n} \log J_n
\label{rate2}
\end{equation} is called an achievable rate if
\begin{eqnarray}
p_e\left(\cC_n(W^{BE})\right) &\rightarrow& 0 \quad {\hbox{as}} \,\, n
\rightarrow \infty \\
\forall i,\ \ \Delta(\cS_i) &\rightarrow& 0 \quad {\hbox{as}} \,\, n
\rightarrow \infty.
\end{eqnarray} We will refer to a sequence of codes simply
as a code when there is no possibility of ambiguity.

%The number $R$ is called an \emph{achievable rate} if for any $\eps,\delta>0$ %and sufficiently large $n$, there exists an $(n,P-\delta,\eps)$ private code $%\cC(W^{BE})$.

It has been shown that \cite{Devetak03} for any classical-quantum channel
$W^{BE}$, and any probability distribution $\bp$ on $\cX$, there exists a
private code $\cC(W^{BE})$ (more precisely, a sequence $\cC(W^{BE}):=
\{\cC_n(W^{BE})\}_{n=1}^\infty$ of private codes) with achievable rate
\begin{equation}\label{EQ_private_rate}
I_c(\bp,W^{BE}):=\chi(\bp,W^B)-\chi(\bp,W^E).
\end{equation}
The private code $\cC(W^{BE})$ constructed by Devetak \cite{Devetak03} requires
knowledge of the channel $W^{BE}$.

Our main result in this paper is to show that one can construct a private code
even without the full knowledge of the $c\rightarrow qq$ channel $W^{BE}$. The
only prior knowledge required for the code construction is that of the
probability distribution on the input alphabet $\cX$, and bounds on the
corresponding Holevo quantities $\chi(\bp,W^B)$ and $\chi(\bp,W^E)$.

\begin{theorem} \label{U_private}
Let  $W^{BE}$ denote a classical-quantum channel with input alphabet $\cX$.
Given a probability distribution $\bp$ on $\cX$ and
positive numbers $\chi_0$ and $\chi_1$, such that $\chi_0 \le \chi(\bp, W^B)$
and  $\chi_1 \ge \chi(\bp, W^E)$, there exists a universal private code
$\cC(W^{BE})$ which can achieve any rate $R\leq \chi_0 - \chi_1$.
% whereI_c(\bp,W^{BE})$.
\end{theorem}
\begin{proof}
The idea of the proof is to combine the universal packing lemma and the universal covering lemma of the previous two sections.

Assume that $\chi_0 > \chi_1$ and define $I_c := \chi_0 - \chi_1$. Note that
$I_c \le I_c(\bp,W^{BE})$, where $I_c(\bp,W^{BE})$ is given by
(\ref{EQ_private_rate}). Let $\{M_n\}_{n=1}^\infty$ and $\{L_n\}_{n=1}^\infty$
denote sequences of positive integers such that \bea
%N&=& 2^{n[I_c^0-\zeta_n((k+2)d_B+d_E)-4c\delta-c(k,\delta)-c(k,3\delta)]}\nonumber\\
\lim_{n \rightarrow \infty} \frac{1}{n}\log M_n &=& \chi_0,\label{enn2}\\
\lim_{n \rightarrow \infty} \frac{1}{n}\log L_n  &=& \chi_1 +2c\delta.
\eea
%{\bf The function $f(\delta,n)$ should be changed according to the code rate of the packing.}
%{\bf{(Is it necessary or helpful to the reader to say the rest of the sentence :
%"where $N$ and $L$ are the size of the c-q code in Sec.~\ref{packing} and the size of the covering set in Sec.~\ref{covering},
%respectively." I think it's better to just define $N$ and $L$ as in the equations above, because their significances
%as sizes of codes or covering sets becomes evident in the course of the proof. What do you think?)}}
%{\bf{I think it would be better to say the following:
In the course of the proof it will become evident that $M_n$ is the size of a
c-q code and $L_n$ the size of a covering set. %}}

Let $J_n = \lfloor M_n/L_n \rfloor $, and define the sets $\cJ_n:=\{1,2,\ldots,
J_n\}$ and $\cL_n:=\{1,2,\ldots, L_n\}$. Fix $\delta, \eps >0$. We first
construct a random code $\cC_{\text{rc}}^n$ whose codewords are given by
realizations of the random variables in the set
$\cA_n:=\{X_{j,\ell}\}_{j\in[\cJ_n],\ell\in[\cL_n]}$, each $X_{j,\ell}$ being a
random variable chosen independently from the $\delta$-typical set,
$\cT_{\bp,\delta}^n$, according to
\begin{equation}\label{EQ_pro_M3}
p'^n_{x^n}:=\Pr\{X_{j,\ell}=x^n\}=\begin{cases}
p^n_{x^n}/Q_n,& \text{if}\ x^n\in\cT_{\bp,\delta}^n,  \\ 0,&
\text{otherwise,}\end{cases}
\end{equation}
where $Q_n$ is defined in \reff{EQ_CT_1}.
%{\bf{In other words, %Alice maps each index $j \in \cP$ to a pair of classical indices $(j,l)$, with
%$j \in \cJ$ and $l \in \cL$.
%the codeword for Alice's $j^{th}$ private message is then given by a realizati% on of the random variable $M_{j,\ell}$.}}

Note that there are at most $M_n$ pairs of classical indices $(j,\ell)$. Then
by using \reff{enn2} and the fact that $\chi_0 \le \chi(\bp, W^B)$, and by
invoking Theorem~\ref{U_packing}, we know that there exists a POVM
$\{\Upsilon_{j,\ell}\}$ defined by (\ref{POVM_packing}), which would allow Bob
to identify the pair of classical indices $(j,\ell)$ with average error
probability $\bbE\left[p_e(\cC^n_{\text{rc}})\right] < \eps_n$, where $\eps_n$
is given by (\ref{e16}) and vanishes asymptotically with $n$. The definition of
the POVM elements only depends on the value of $\chi_0$ and not on specific
knowledge of the structure of the channel $W^B$.

Define the event
$$\bI_0\equiv \bI_0^{(n)}:=\left\{p_e(\cC^n_{\text{rc}})\leq \sqrt{\eps_n}\right\},$$
and let ${\bI_0^c}$ denote the complement of the event $\bI_0$.
Then using the upper bound on $\bbE\left[p_e(\cC^n_{\text{rc}})\right]$,
and the Markov inequality (see e.g. \cite{Grimmett_Stirzaker01})
for the random variable $p_e(\cC^n_{\text{rc}})$, we have that
\bea
\Pr\{\bI_0^c\} &=& \Pr\{p_e(\cC^n_{\text{rc}}) > \sqrt{\eps_n}\} \nonumber \\
& \leq &
\frac{\bbE\left[p_e(\cC^n_{\text{rc}})\right]}{\sqrt{\eps_n}} \nonumber\\
&\leq & \sqrt{\eps_n}.
\label{EQ_ent0}
\eea

Let $\cA_j^{(n)}=\{X_{j,\ell}\}_{\ell\in\cL_n}$, and define the following
events:
$$\bI_j:=\left\{\Delta(\cA_j^{(n)})< \eps+4\sqrt{k\eps}+8\sqrt{3\eps+2\sqrt{k\eps}}\right\},$$
where $\Delta(\cA_j^{(n)})$ denotes the ``obfuscation error'' of the set
$\cA_j^{(n)}$, defined as in \reff{EQ_Delta}.
 Then by invoking Theorem~\ref{U_covering}, we have for each $j=1,\ldots,J_n$
\begin{equation}\label{EQ_entk}
\Pr\{\bI_j^c\}\leq \eps_n^\prime,
\end{equation}
for a positive constant $\eps_n^\prime$ defined by (\ref{EQ_epn_n}),
which vanishes asymptotically with $n$.

Combining (\ref{EQ_ent0}) and $(\ref{EQ_entk})$ gives \bea \Pr\{(\bI_0\cap
\bI_1 \cap\cdots\cap\bI_{J_n})^c\} &=&
\Pr\{\bI_0^c \cup \bI_1^c \cup \cdots \cup \bI_{J_n}^c\} \nonumber\\
&\leq & \sum_{j=0}^{J_n} \Pr\{\bI_j^c\} \leq J_n \eps_n^\prime +\sqrt{\eps_n},\nonumber\\
\label{end} \eea which goes to 0 as $n\to\infty$. Hence, there exists at least
one realization of the set $\cA_n$ (and hence of the random code
$\cC^n_{\text{rc}}$), for which each of the events $\bI_j$, for $j=0,1,\ldots,
J_n$, occurs. Equivalently, (\ref{EQ_cond_p1}) and (\ref{EQ_cond_p2}) hold.

Note that (\ref{end}) implies that for $n$ large enough, the event $\bI_0\cap
\bI_1 \cap\cdots\cap\bI_{J_n}$ occurs with high probability, and this in turn
ensures that the realizations of the sets $\cA_j^{(n)}$ for $j \in \cJ_n$ are
mutually disjoint. Hence, in order to reliably transmit private classical
messages to Bob, Alice maps each of her private messages onto a randomly picked
element in each different disjoint set, there being $|\cJ_n| = J_n$ such sets.
She can thus achieve a rate
$$R=\lim_{n\rightarrow \infty}\frac{1}{n}\log J_n\le I_c -2 c\delta,$$
%where $$f'(\delta,n)=f(\delta,n)+2c\delta \to 0$$
%as $n\to \infty$ and $\delta\to 0$. {\bf The function $f'(\delta,n)$ should al%so be changed according to the code rate of the packing.}
%
%$P \leq J_0$,
%where $J_0 :=\frac{1}{n}\log J=I_c^0-\zeta_n((k+2)d_B+d_E)-4c\delta-c(k,\delta%)-c(k,3\delta)$.
Since $\delta$ is arbitrary, we arrive at the statement $R \rightarrow I_c=\chi_0-\chi_1$ from below, and there exists at least one particular realization $\cC_n=\{x^n_i,\Upsilon_i\}_{i\in \cJ_n}$ of $\cC^n_{\text{rc}}$ of rate $R$ such that
    $$p_e(\cC_n)\leq\sqrt{\eps_n}.$$

\end{proof}

%Finally, we can apply the following arguments.
%\begin{enumerate}
%\item Derandomization: There exists at least one particular realization
%    $\cC_n=\{x^n_i,\Upsilon_i\}_{i\in \cJ_n}$ of $\cC^n_{\text{rc}}$ such that
%    $$p_e(\cC_n)\leq\sqrt{\eps_n}.$$
%\item %Average to Maximal Error Probability: At least $J_n/2$ pairs of the
%%    encoding and decoding POVM in $\cC^*_n$ contribute less than twice of the
%%    expected error probability. We can throw away the other half of the
%%    code $\cC^*_n$, and define $\cC_n$ to be the remaining set. This guarantees
%%    that the maximal error probability of $\cC_n$ is upper bounded by
%%    $$p_e(\cC_n)\leq 2\sqrt{\eps_n}.$$
%    Obviously, the sequence of codes $\{\cC_n\}_{n=1}^\infty$ also has a rate $R$, and the probability of error can be made arbitrarily small for {\em{all}}
%codeowrds as $n$ becomes large. Hence any rate $R< I_c$ is achievable.
%%However, the rate of $\cC$ now becomes $R':=R-\frac{1}{n}$. Overall, $R'\to
%%    I_c$ from below as $n\to \infty$.
%\end{enumerate}

%%% end insert %%%%%%%%
\section{Conclusion}\label{conclusion}
We have shown that there exists a universal private coding scheme for a
$c\rightarrow qq$ channel, $W^{BE}$, which Alice can use to transmit private
messages to Bob (who has access to the system $B$), at the same time
ensuring that an eavesdropper, Eve (who has access to the system $E$) does
not get any information about her messages. The coding scheme only requires
knowledge of the probability distribution $\bp$ on the input alphabet, $\cX$,
of the channel, and of the corresponding
bounds on the Holevo quantities of the
$c\rightarrow q$ channels $W^B$ and $W^E$, induced between Alice and Bob, and
Alice and Eve, respectively. More precisely, a lower bound (say $\chi_0$)
%$\chi_0$
on $\chi(\bp,W^B)$, and an upper bound (say $\chi_1$) on $\chi(\bp,W^E)$ suffices.
This is because
%The requirement of the knowledge of these bounds follows from the fact that
our universal private coding scheme is obtained by combining
%two important lemmas:
the universal packing lemma (of Section \ref{packing}) and the
universal covering lemma (of Section \ref{covering}), which depend on these bounds.

Prior knowledge of the probability distribution $\bp$ on the channel's
input alphabet seems crucial both in Hayashi's coding scheme \cite{Hayashi09CMP},
as well as in ours. In contrast, the classical results on
 universal packing and universal covering \cite{CK81} do not require this knowledge.
%the dependence on the knowledge of the prior probability distribution $\bp$ on the set of classical message $\cX$ seems crucial.
For the universal private coding, we also require knowledge of the values of $\chi_0$ and $\chi_1$ individually.
%instead of one combined value $I_c:=\chi_0-\chi_1$.
%for the universal private coding due to direct combination of universal packing and universal covering techniques in the proof.
It would be interesting to know whether there exists other coding schemes which require less prior knowledge,
for example knowledge of $I_c:=\chi_0-\chi_1$ alone.

Our result on universal private coding can be generalized to the case of
general quantum channels if some further information is available, as explained
below. In the most general setting of transmitting private information over a
quantum channel, $\cN$, the sender Alice prepares a quantum input state $\rho$.
Notice that there are unlimited number of ensemble decompositions
$\{p_i,\rho_i\}$ of $\rho$, satisfying $\rho=\sum_{i\in\cX}p_i \rho_i$. Each
such decomposition
%$\{p_i,\rho_i\}_{i\in\cX}$ of a input state $\rho$
induces a $c\rightarrow qq$ channel, $W^{BE}$, from the isometric extension $U_{\cN}$ of $\cN$, such that the induced channels
are given by $W^B:i\rightarrow \cN(\rho_i)$ and
$W^E:i\rightarrow \widehat{\cN}(\rho_i)$, where $\widehat{\cN}$ denotes the channel which is complementary to $\cN$.
By using Theorem \ref{U_private}, for each such ensemble $\{p_i,\rho_i\}_{i\in\cX}$, we can then design a
universal coding scheme which achieves a
private transmission rate equal to $I_c(\{p_i,\rho_i\})$.
Moreover, if prior knowledge of an
ensemble which would lower bound the Holevo quantities of all possible $c\rightarrow q$ channels $W^B$, and
upper bound the Holevo quantities for all possible $c\rightarrow q$ channels $W^E$, is available, then we can achieve
the universal private coding for this general setting by direct application of Theorem \ref{U_private}.
The private transmission rate achieved would then be given by $$I_p:=\min_{\{p_i,\rho_i\}} I_c(\{p_i,\rho_i\}).$$

%The result in \cite{Hayashi09CMP} on the universal coding for the $(c\to q)$ channel $W^B$ cannot be generalized to the universal private coding directly from combining with the covering technique we used in Sec.~\ref{private}, due to its construction of the code from a certain type in $\cX^n$. In order for the covering set to satisfy the criterion (\ref{EQ_cond_p2}), it must be constructed from the $\delta-$typical set $\cT_{\bp,\delta}^n$ of $\cX^n$ according to Theorem~\ref{U_covering}. Another open question would be to generalize the type covering proved in \cite{CK81} to the quantum domain.
%

\section*{Acknowledgments}

We would like to thank Milan Mosonyi and Mark Wilde
for their helpful comments.
ND is also grateful to Ismail Akhalwaya for
doing some useful numerics in relation to an earlier version
of this paper. The research leading to these results has received
funding from the European Community's Seventh Framework Programme
(FP7/2007-2013) under grant agreement number 213681.

\appendix
\section{Proof of (\ref{eq_dim})}\label{app_dim}
Let $M_d$ be the linear space of all $d\times d$ complex matrices. Then
$\tilde{\cK}_\bq\subseteq \Span\{A^{\otimes n}\ket{y^n}: A\in M_d, \ket{y^n}\in\cK_\bq\}$. For any fixed $\ket{y^n}\in\cK_\bq$, let
$$\cK_\bq(y^n)=\Span\{A^{\otimes n}\ket{y^n}:A\in M_d\}.$$
It then follows from Ref.~\cite{JHHH98PRL} that
\[
|\cK_\bq(y^n)|\leq (n+1)^{d^2}.
\]
By dimension counting, we have
\[
|\tilde{\cK}_\bq|\leq \sum_{y^n \in\cT_{\cY}^n(\bq)}|\cK_\bq(y^n)| \leq (n+1)^{d^2}|\cK_\bq|,
\]
where $|\cT_{\cY}^n(\bq)|=|\cK_\bq|$
\section{Commutation relations}\label{commu}
\begin{lemma}
For any $\sigma\in\cD(\sH)$, and $\tau_n$ defined by \reff{EQ_tau_n},
the commutator $[\tau_{n},\sigma^{\otimes n}]=0$.
\end{lemma}
\begin{proof}
Let $\cY=\{1,2,\cdots,d\}$ and $d=\dim\sH$.  Recall that
$$\sigma^{\otimes n}=\sum_{\bq\in \cP^n_\cY}
2^{-n\left[D(\bq\|\underline{\lambda})+H(\bq)\right]} I_\bq,$$ where $\underline{\lambda}=(\lambda_1,\cdots,\lambda_d)$ is the vector of eigenvalues of $\sigma$ and $I_\bq$ is
the projection operator defined by \reff{iq}.

We have $I_\bq\leq \tilde{I}_\bq$ from the fact that $\cK_\bq\subseteq\tilde{\cK}_\bq$, $\forall\bq\in \cP^n_\cY$. It follows trivially that
$[\tau_{n},\sigma^{\otimes n}]=0$.
\end{proof}

\begin{lemma} \label{LEMMA_commutator2}
Given a type $\bp\in \cP^n_\cX$, and a sequence $x^n \in \cT_\cX^n(\underline{p})$, let $x_o^n$ be the corresponding ordered sequence in
$\cT_\cX^n(\underline{p})$. Then we have  $[\omega_{x_o^n},\tau_{n}]=0$.
\end{lemma}
\begin{proof} By definitions of (\ref{EQ_tau_n}) and (\ref{EQ_omega_o}),

$$\omega_{x_o^n}=\left(\prod_{i=1}^k \frac{1}{|\cP^{m_i}_\cY|}\right)\sum_{\bq^{(\x_1)}\in
\cP^{m_1}_\cY}\cdots\sum_{\bq^{(\x_k)}\in
\cP^{m_k}_\cY}\tau_{\bq^{(\x_1)}}\otimes\cdots\otimes\tau_{\bq^{(\x_k)}}.$$
It follows from (\ref{EQ_tau_q}) that the term
$\tau_{\bq^{(\x_1)}}\otimes\cdots\otimes\tau_{\bq^{(\x_k)}}$ in the summand of
$\omega_{x_o^n}$ can be written as:
$$\tau_{\bq^{(\x_1)}}\otimes\cdots\otimes\tau_{\bq^{(\x_k)}}=\frac{\tilde{I}_{\bq^{(\x_1)}}
\otimes\cdots\otimes \tilde{I}_{\bq^{(\x_k)}}}{\prod_{i=1}^k
|\cK^{m_i}_{\bq^{(\x_i)}})|},$$ where each
$\bq^{(\x_i)}=(q^{(\x_i)}_1,\cdots,q^{(\x_i)}_d)$ is a type in $\cP^{m_i}_\cY$.
The following proof holds for every term of $\omega_{x_o^n}$. Define $I^*
:= I_{\bq^{(\x_1)}}\otimes\cdots\otimes I_{\bq^{(\x_k)}}$, where
\begin{equation}\label{EQ_I_star}
I^*= \Bigl(\sum_{y^{m_1}\in\cT_\cY^{m_1}({\bq^{(\x_1)}})}\proj{y^{m_1}}\Bigr)\otimes\cdots  \otimes  \Bigl(\sum_{y^{m_k}\in\cT_\cY^{m_k}({\bq^{(\x_k)}})}\proj{y^{m_k}}\Bigr).
\end{equation}
For any sequence $y^n
\in\cT^{m_1}_\cY({\bq^{(\x_1)}})\times\cdots\times\cT^{m_k}_\cY({\bq^{(\x_k)}})$,
the number of times any $i\in\cY$ appears in the sequence $y^n$ is given by
$$N(i|y^n)= \sum_{j=1}^k m_j q_i^{(\x_j)},$$
where $q_i^{(\x_j)}$ is the $i^{th}$ element in the probability vector
$\bq^{(\x_j)}$. Such a sequence $y^n$ must also belong to $\cT_\cY^n(\bq)$ for
the type $\bq\in \cP^n_\cY$, where the $i$-th element of $\bq$ is
$$q_i=\frac{N(i|y^n)}{n}=\frac{\sum_{j=1}^k n p_j q^{(\x_j)}_i}{n}=\sum_{j=1}^k p_j q_i^{(\x_j)}.$$
In short, we can write $\bq=\sum_{j=1}^k p_j\bq^{(\x_j)}$. Therefore
$\cT^{m_1}_\cY({\bq^{(\x_1)}})\times\cdots\times\cT^{m_k}_\cY({\bq^{(\x_k)}})
\subset \cT^n_\cY({\bq})$, and $$I^* \leq I_\bq.$$ Furthermore, we can obtain
$$\tilde{I}^* \leq \tilde{I}_\bq.$$
Therefore
$[\tau_{\bq^{(\x_1)}}\otimes\cdots\otimes\tau_{\bq^{(\x_k)}},\tau_{\bq}]=0$ for
all $\bq\in \cP^n_\cY$  and for all $\bq^{(\x_i)}\in \cP_{\cY}^{m_i}$. By
linearity, we have $[\omega_{x_o^n},\tau_{n}]=0$.
\end{proof}

\begin{lemma}
$[\omega_{x^n}, \tau_{n}]=0$.
\end{lemma}
\begin{proof}
For any sequence $y^n
\in\cT^{m_1}_\cY({\bq^{(\x_1)}})\times\cdots\times\cT^{m_k}_\cY({\bq^{(\x_k)}})$, we know from Lemma~\ref{LEMMA_commutator2} that
$y^n\in\cT^n_\cY(\bq)$, where $\bq=\sum_{j=1}^k p_j\bq^{(\x_j)}$. Furthermore,
$sy^n$ must also belong to $\cT_\cY^n(\bq)$, where $s\in S_n$ is a permutation such that $x^n=s x_o^n$. Denote $$\cJ^n=\{sy^n: \forall
y^n\in\cT^{m_1}_\cY({\bq^{(\x_1)}})\times\cdots\times\cT^{m_k}_\cY({\bq^{(\x_k)}})\}\subset\cT^n_\cY(\bq).$$
Let $I^*_s=U_sI^*U_s^\dagger$, where $I^*$ is defined in (\ref{EQ_I_star}). Obviously $I^*_s\leq I_\bq$.
Then following the same argument as in Lemma~\ref{LEMMA_commutator2}, we can conclude that $[\omega_{x^n},\tau_{n}]=0$.

\end{proof}

%\bibliographystyle{aipsamp}

%\bibliography{Ref}

\end{document}